\documentclass{amsart}

\usepackage{amsfonts}

\usepackage{url}

\providecommand{\middle}{\bigg}

\theoremstyle{plain}
\newtheorem{Thm}{Theorem}

\newtheorem{Lem}{Lemma}
\theoremstyle{definition}
\newtheorem{Rem}{Remark}

\renewcommand{\Re}{\operatorname{Re}}

\allowdisplaybreaks[3]

\begin{document}

\title{A $q$-linear analogue of the plane wave expansion}

\author[L. D. Abreu]{Lu\'{\i}s Daniel Abreu}
\address{CMUC, Departamento de Matem\'atica, Universidade de Coimbra,
Faculdade de Ci\^{e}ncias e Tecnologia (FCTUC), 3001-454 Coimbra, Portugal}
\curraddr{Austrian Academy of Sciences, Acoustic Research Institute, 
Reichsratsstrasse 17, A-1010 Vienna, Austria}
\email{daniel@mat.uc.pt}
\thanks{Research of the first author supported by CMUC/FCT and FCT
post-doctoral grant SFRH/BPD/26078/2005, POCI 2010 and FSE}
\author[\'O. Ciaurri]{\'Oscar Ciaurri}
\address{Departamento de Matem\'aticas y Computaci\'on, Universidad de La
Rioja, 26004 Logro\~no, Spain}
\email{oscar.ciaurri@unirioja.es}
\author[J. L. Varona]{Juan Luis Varona}
\address{Departamento de Matem\'aticas y Computaci\'on, Universidad de La
Rioja, 26004 Logro\~no, Spain}
\email{jvarona@unirioja.es}
\urladdr{http://www.unirioja.es/cu/jvarona/}

\renewcommand{\datename}{\textbf{TO APPEAR IN}:}
\date{\textit{Adv.\@ in Appl.\@ Math.\@} (accepted in November 2012)}

\thanks{Research of the second and third authors supported by grant
MTM2012-36732-C03-02 of the DGI}



\begin{abstract}
We obtain a $q$-linear analogue of Gegenbauer's expansion of the
plane wave. It is expanded in terms of the little $q$-Gegenbauer
polynomials and the \textit{third} Jackson $q$-Bessel function.
The result is obtained by using a method based on bilinear
biorthogonal expansions.
\end{abstract}

\maketitle

\section{Introduction}

Let $\beta >-1/2$. Gegenbauer's expansion of the plane wave in
Gegenbauer polynomials and Bessel functions is
\begin{equation}
e^{ixt} = \Gamma (\beta )\left( \frac{x}{2}\right) ^{-\beta }
\sum_{n=0}^{\infty }i^{n}(\beta +n)
J_{\beta +n}(x)C_{n}^{\beta}(t),
\quad t\in [-1,1]  \label{eq:geg}
\end{equation}
(see \cite[Ch.~11, \S\,5, formula (2)]{Wat}). In \cite[formula
(3.32)]{IZ}, Ismail and Zhang have discovered a basic analogue
of~\eqref{eq:geg} on $q$-quadratic grids, expanding the
$q$-quadratic exponential function (a solution of a first order
equation involving the so-called Askey-Wilson operator) in terms
of \emph{second} Jackson $q$-Bessel functions, $J_{\nu}^{(2)}(z;q)$,
and the continuous $q$-Gegenbauer polynomials (moreover, this has
been later extended to continuous $q$-Jacobi polynomials
in~\cite{IsmRZ}). Their $q$-quadratic exponential inherits the
orthogonality of the continuous $q$-Gegenbauer polynomials, and
leads to a theory of Fourier series on $q$-quadratic grids
(see~\cite{BS} and~\cite{Sus}). Since then, it has been a folk open
question to find a $q$-analogue of~\eqref{eq:geg} involving
discrete $q$-Gegenbauer polynomials and $q$-Bessel functions of a
different type. In this note, and also for $\beta >-1/2$, we will
obtain the following $q$-analogue of~\eqref{eq:geg}:
\begin{equation}
e(itx;q^{2}) =
\frac{(q^{2};q^{2})_{\infty }}{(q^{2\beta};q^{2})_{\infty }}
\, x^{-\beta } \sum_{n=0}^{\infty } i^{n}
q^{-[\frac{n+1}{2}](\beta-\frac{1}{2})} (1-q^{2\beta+2n })
J_{\beta+n}^{(3)}(xq^{[\frac{n+1}{2}]};q^{2})
C_{n}^{\beta }(t;q^{2}), \label{eq:qlinearplanewave}
\end{equation}
for $t\in [-1,1]$. Here, $[r]$ denotes the biggest integer less or
equal than~$r$, $J_{\nu }^{(3)}(z;q)$ is the \emph{third} Jackson
$q$-Bessel function, and $C_{n}^{\beta }(t;q^{2})$ are
$q$-analogues of the Gegenbauer polynomials, defined in terms of
the little $q$-Jacobi polynomials (see definitions of all these
functions in the third section of the paper). In fact we prove a
more general formula than~\eqref{eq:qlinearplanewave},
see Theorem~\ref{thm:qDunkl}.

The $q$-exponential function in~\eqref{eq:qlinearplanewave} is
the one introduced in~\cite{Rubin}:
\begin{equation*}
e(z;q^{2}) = \cos (-iz;q^{2}) + i\sin (-iz;q^2),
\end{equation*}
where
\begin{equation*}
\cos (z;q^{2}) =
\frac{(q^{2};q^{2})_{\infty }}{(q;q^{2})_{\infty }}
z^{\frac{1}{2}} J_{-\frac{1}{2}}^{(3)}(z;q^{2})
\quad\text{and}\quad
\sin (z;q^{2}) =
\frac{(q^{2};q^{2})_{\infty }}{(q;q^{2})_{\infty }}
z^{\frac{1}{2}} J_{\frac{1}{2}}^{(3)}(z;q^{2}).
\end{equation*}
The expansion \eqref{eq:qlinearplanewave} is obtained as a special case of a
more general formula, which is a $q$-analogue of the expansion of the Dunkl
kernel in terms of Bessel functions and generalized Gegenbauer polynomials
(see \cite{ACVExp} and \cite{Rosler}). The technique of proof is based on
the method of Bilinear Biorthogonal Expansions, developed in \cite{ACVExp}
and provides as a byproduct, $q$-analogues of Neumann series which are valid
for functions which belong to certain $q$-analogues of the Paley-Wiener
space. We remark that another $q$-linear analogue (but involving completely
different functions) of~\eqref{eq:geg} has been obtained in~\cite{ISem}.

The paper is organized as follows. We describe the setup of the method of
Bilinear Biorthogonal Expansions in the next section. In the third section
we collect some material on basic hypergeometric functions and
$q$-integration and apply it in the fourth section to the context of
the general set-up, yielding our main result. The last section
contains the evaluation of some $q$-integrals which are essential
in the proofs.

\section{The method of Bilinear Biorthogonal Expansions}
\label{sec:BBE}

We proceed to describe the set-up of the method of Bilinear Biorthogonal
Expansions~\cite{ACVExp}. The method aims to finding a bilinear expansion
for $K(x,t)$, a function of two variables defined on $\Omega \times \Omega
\subset \mathbb{R}\times \mathbb{R}$ and such that $K(x,t) = K(t,x)$
almost everywhere for $(x,t) \in \Omega \times \Omega$.
It consists of three ingredients:

\begin{enumerate}
\item[(i)] First define on $L^{2}(\Omega ,d\mu)$, with $d\mu$ a non-negative
real measure, an integral transformation $\mathcal{K}$ with inverse
$\widetilde{\mathcal{K}}$,
\begin{equation*}
(\mathcal{K}f)(t) = \int_{\Omega }f(x)\overline{K(x,t)}\,d\mu (x),
\qquad
(\widetilde{\mathcal{K}}g)(x) = \int_{\Omega }g(t)K(x,t)\,d\mu(t).
\end{equation*}
As usual, it is enough to suppose that the operators
$\mathcal{K}$ and $\widetilde{\mathcal{K}}$ are defined with these
formulas on a suitable dense subset of $L^2(\Omega,d\mu)$, and later
extended to the whole $L^2(\Omega,d\mu)$ in the standard way.
Let also note that, by Fubini's theorem, they satisfy
\begin{equation*}
\int_{\Omega }(\mathcal{K}f)g\,d\mu = \int_{\Omega }(\mathcal{K}g)f\,d\mu ,
\qquad
\int_{\Omega }(\widetilde{\mathcal{K}}f)g\,d\mu
= \int_{\Omega }(\widetilde{\mathcal{K}}g)f\,d\mu .
\end{equation*}

\item[(ii)] Let $I\subset \Omega $ be an interval such that, as a function of~$t$,
$K(x,\cdot )\in L^{2}(I,d\mu )$ and consider the subspace $\mathcal{P}$ of
$L^{2}(\Omega ,d\mu )$ constituted by those functions $f$ such that
$\mathcal{K}f$ vanishes outside of~$I$. This can also be written as
\begin{equation*}
\mathcal{P} = \Bigl\{f\in L^{2}(\Omega ) : 
f(x) = \int_{I}u(t)K(x,t)\,d\mu (t),\ u\in L^{2}(I,d\mu )\Bigr\}.
\end{equation*}

\item[(iii)] Finally, consider a pair of complete biorthonormal sequences of
functions in $L^{2}(I,d\mu )$, $\{P_{n}\}_{n\in N}$ and $\{Q_{n}\}_{n\in N}$
(with $N$ a subset of~$\mathbb{Z}$) and define, in $L^{2}(\Omega ,d\mu )$,
the sequences of functions $\{S_{n}\}_{n\in N}$ and $\{T_{n}\}_{n\in N}$
given by
\begin{equation*}
S_{n}(x)=\widetilde{\mathcal{K}}(\chi _{I}\overline{Q_{n}})(x),
\quad x\in \Omega ,
\qquad
T_{n}(x)=\overline{\mathcal{K}(\chi _{I}P_{n})(x)},
\quad x\in \Omega
\end{equation*}
(note that if $P_{n}=Q_{n}$ then $S_{n}=T_{n}$).
\end{enumerate}

Then, the following holds (see \cite[Theorem~1]{ACVExp}):

\begin{Thm}
\label{thm:expbilin} For each $x\in \Omega $, the following expansion
\footnote{%
The condition $t\in I$ in the identity~\eqref{eq:expbilin} is \textit{not} a
mistake. Although $K(x,t)$ is defined on $\Omega \times \Omega $, the
functions $P_{n}(t)$ are defined, in general, only on~$I$.} holds
in $L^{2}(I,d\mu )$:
\begin{equation}
K(x,t) = \sum_{n\in N}P_{n}(t)S_{n}(x),\quad t\in I.
\label{eq:expbilin}
\end{equation}
Moreover, $\{S_{n}\}_{n\in N}$ and $\{T_{n}\}_{n\in N}$ are a pair of
complete biorthogonal sequences in~$\mathcal{P}$, in such a way that every
$f\in \mathcal{P}$ can be written as
\begin{equation*}
f(x) = \sum_{n\in N}c_{n}(f)S_{n}(x),\quad x\in \Omega ,
\end{equation*}
with
\begin{equation*}
c_{n}(f) = \int_{\Omega }f(t)\overline{T_{n}(t)}\,d\mu (t).
\end{equation*}
The convergence is uniform in every set where
$\left\Vert K(x,\cdot)\right\Vert _{L^{2}(I,d\mu )}$ is bounded.
\end{Thm}

\section{Preliminaries on $q$-special functions}

We follow the standard notations (see~\cite{GR}
and~\cite{KoeLeSw}). Choose a number $q$ such that $0<q<1$. The
notational conventions
\begin{gather*}
(a;q)_{0}=1,\quad (a;q)_{n}=\prod_{k=1}^{n}(1-aq^{k-1}), \\
(a;q)_{\infty } = \lim\limits_{n\rightarrow \infty }(a;q)_{n},\quad
(a_{1},\dots ,a_{m};q)_{n}=\prod_{l=1}^{m}(a_{l};q)_{n},\quad |q|<1,
\end{gather*}
where $n=1,2,\dots $ will be used.
The symbol ${}_{s}\phi _{r}$ stands for the function
\begin{equation*}
{}_{s}\phi _{r}\left( \,
\begin{matrix}
a_{1},\dots ,a_{s} \\
b_{1},\dots ,b_{r}
\end{matrix}
\,\middle|\,q;z\,\right)
= \sum_{n=0}^{\infty }
\frac{ \left((-1)^n q^{n(n-1)/2}\right)^{r-s+1} (a_{1},\dots,a_{s};q)_{n}}
{(q,b_{1},\dots ,b_{r};q)_{n}}z^{n};
\end{equation*}
in particular,
\begin{equation*}
{}_{r+1}\phi _{r}\left( \,
\begin{matrix}
a_{1},\dots ,a_{r+1} \\
b_{1},\dots ,b_{r}
\end{matrix}
\,\middle|\,q;z\,\right)
= \sum_{n=0}^{\infty } \frac{(a_{1},\dots,a_{r+1};q)_{n}}
{(q,b_{1},\dots ,b_{r};q)_{n}}z^{n}.
\end{equation*}
We will also require the definition of the $q$-integral. The $q$-integral in
the interval $(0,a]$ is defined as
\begin{equation*}
\int_{0}^{a}f(t)\,d_{q}t = (1-q)a\sum_{n=0}^{\infty }f(aq^{n})q^{n}
\end{equation*}
and in the interval $(0,\infty )$ as
\begin{equation*}
\int_{0}^{\infty }f(t)\,d_{q}t
= (1-q)\sum_{n=-\infty }^{\infty } f(q^{n})q^{n},
\end{equation*}
provided that the infinite sums converge absolutely.
This can be extended to the whole real line in an obvious way.

The \emph{third Jackson $q$-Bessel function} $J_{\nu }^{(3)}$ is defined by
the power series
\begin{align*}
J_{\nu }^{(3)}(x;q)
&= \frac{(q^{\nu +1};q)_{\infty }}{(q;q)_{\infty }} \,x^{\nu}
\sum_{n=0}^{\infty } (-1) ^{n}
\frac{q^{\frac{n(n+1)}{2}}}{(q^{\nu +1};q)_{n}(q;q)_{n}} \,x^{2n}
\\
&= \frac{(q^{\nu +1};q)_{\infty }}{(q;q)_{\infty }} \,x^{\nu}
\,{}_{1}\phi _{1}\left( \,
\begin{matrix}
0 \\
q^{\nu+1}
\end{matrix}
\,\middle|\,q; qx^2 \,\right).
\end{align*}
Throughout this paper, when no confusion is possible, we will drop the
superscript and write simply
\begin{equation*}
J_{\nu }(x;q) = J_{\nu }^{(3)}(x;q).
\end{equation*}

For $x\in (0,1)$, the \emph{little $q$-Jacobi} polynomials are defined for
$\alpha,\beta >-1$ by
\begin{equation*}
p_{n}(x;q^{\alpha},q^{\beta};q) = {_{2}\phi_{1}} \left(\,
\begin{matrix}
q^{-n},q^{\alpha+\beta+n+1} \\
q^{\alpha+1}
\end{matrix}
\,\middle|\, q;qx\,\right).
\end{equation*}
They satisfy the following discrete orthogonality relation
(see~\cite[(14.12.2)]{KoeLeSw}):
\begin{multline*}  
\int_{0}^{1}\frac{(qx;q)_{\infty}} {(q^{\beta+1}x;q)_{\infty}} \,x^{\alpha}
p_{n}(x;q^{\alpha },q^{\beta};q) p_{m}(x;q^{\alpha},q^{\beta};q) \, d_{q}x \\
= \frac{(1-q)(1-q^{\alpha+\beta+1})}{1-q^{\alpha+\beta+2n+1}}
\frac{(q,q^{\alpha+\beta+2};q)_{\infty}}
{(q^{\alpha+1},q^{\beta+1};q)_{\infty}}
\frac{(q,q^{\beta+1};q)_{n}}
{(q^{\alpha+1},q^{\alpha+\beta+1};q)_{n}} \, q^{n(\alpha+1)}
\delta_{m,n}.
\end{multline*}
For our purposes we need to rewrite this orthogonality. We will use the
polynomials $p_{n}^{(\alpha,\beta)}$ normalized as follows:
\begin{equation}  \label{eq:pnlittlejacobi}
p_{n}^{(\alpha,\beta)}(x;q) = q^{-\frac{n(\alpha+1)}{2}}
\frac{(q^{\alpha+1};q)_{n}}{(q;q)_{n}}
\, p_{n}(x;q^{\alpha},q^{\beta};q).
\end{equation}
These polynomials satisfy
\begin{equation*}
\lim_{q\rightarrow 1} p_{n}^{(\alpha,\beta)}(x;q) 
= P_{n}^{(\alpha,\beta)}(1-2x),
\end{equation*}
where $P_{n}^{(\alpha,\beta)}$ are the classical Jacobi polynomials (see
\cite[p.~478]{Ism}). It will be convenient to replace $q$ by $q^{2}$ in the
above orthogonality. Then, from the definition of the $q$-integral we obtain
the identity
\begin{equation*}
\int_{0}^{1} f(x) \,d_{q^{2}}x = (1+q) \int_{0}^{1} xf(x^{2})
\,d_{q}x,
\end{equation*}
and use it in order to obtain the following:
\begin{multline}
\int_{0}^{1} \frac{(q^{2}x^{2};q^{2})_{\infty}}
{(q^{2\beta+2}x^{2};q^{2})_{\infty}}
p_{n}^{(\alpha,\beta)}(x^{2};q^{2})
p_{m}^{(\alpha,\beta)}(x^{2};q^{2}) x^{2\alpha+1} \, d_{q}x
\label{eq:ortqjacob} \\
= \frac{1-q}{1-q^{2\alpha+2\beta+4n+2}}
\frac{(q^{2n+2},q^{2\alpha+2\beta+2n+2};q^{2})_{\infty}}
{(q^{2\alpha+2n+2},q^{2\beta+2n+2};q^{2})_{\infty}} \,\delta_{m,n}.
\end{multline}

\section{Application of the method}

\subsection{The integral transform}

We will construct the integral transform required in the first ingredient
of our method. A generalized $q$-exponential kernel (in the spirit of the kernel
for the Dunkl transform) can be defined in terms of $q$-Bessel. Indeed, we
can consider the function
\begin{align}
E_{\alpha }(ix;q^{2})
&= \frac{(q^{2};q^{2})_{\infty }}{(q^{2\alpha+2};q^{2})_{\infty }}
\left( \frac{J_{\alpha }(x;q^{2})}{x^{\alpha }}
+ \frac{J_{\alpha +1}(x;q^{2})}{x^{\alpha +1}}\,xi\right)  \label{eq:qDunkl}
\\
&= {}_{1}\phi _{1}\left( \,
\begin{matrix}
0 \\
q^{2\alpha+2}
\end{matrix}
\,\middle|\,q^2; q^2x^2 \,\right)
+ \frac{ix}{1-q^{2\alpha+2}}
\,{}_{1}\phi _{1}\left( \,
\begin{matrix}
0 \\
q^{2\alpha+4}
\end{matrix}
\,\middle|\,q^2; q^2x^2 \,\right).
\notag
\end{align}
Taking the measure
\begin{equation*}
d\mu _{q,\alpha }(x)
= \frac{1}{2(1-q)}
\frac{(q^{2\alpha +2};q^{2})_{\infty }}{(q^{2};q^{2})_{\infty }}
|x|^{2\alpha +1} \, d_{q}x,
\end{equation*}
in a similar way to the Dunkl transform, first introduced by Dunkl
in \cite{Dunkl} (see also~\cite{Jeu} or~\cite{CV}), for $\alpha
\ge -1/2$ we can define the following $q$-integral transform:
\begin{equation}
\mathcal{F}_{\alpha ,q}f(y)
= \int_{-\infty }^{\infty } f(x)
E_{\alpha}(-iyx;q^{2}) \, d\mu _{q,\alpha }(x),
\quad y\in \{\pm q^{k}\}_{k\in \mathbb{Z}},  \label{eq:qDu-T}
\end{equation}
for $f\in L^{1}(\mathbb{R},d\mu _{q,\alpha })$. This $q$-integral
transform is related to the $q$-Dunkl type operator introduced in
\cite{BeBe} (for a different $q$-Dunkl type operator see
\cite{Fit}). The case $\alpha =-\frac{1}{2}$ provides a
$q$-analogue of the Fourier transformation. In this special case,
an inversion theory of this transform has been derived in~\cite{Rubin}
using the results of \cite{KoorSw, KoorSw-corr}. For even functions
this becomes a $q$-analogue of the Hankel transform
\begin{equation*}
H_{\alpha ,q}f(x)
= \int_{0}^{\infty }
\frac{J_{\alpha }(xy;q^{2})}{(xy)^{\alpha }}
\, f(y)\,d\omega _{q,\alpha }(y),\quad x>0,
\end{equation*}
where $d\omega _{q,\alpha }(y)=\frac{y^{2\alpha
+1}}{1-q}\,d_{q}y$. This is the transform studied by Koornwinder
and Swarttouw \cite{KoorSw, KoorSw-corr}, up to a small
modification. By the results in \cite{KoorSw, KoorSw-corr} we have
the inversion formula
\begin{equation*}
f(q^{n}) = H_{\alpha ,q}(H_{\alpha ,q}f)(q^{n}).
\end{equation*}
In particular, $H_{\alpha ,q}$ is an isometric transformation in
$L^{2}((0,\infty ),d\omega _{q,\alpha })$. For odd functions,
$\mathcal{F}_{\alpha,q}$ turns down to~$H_{\alpha+1,q}$.

For $\mathcal{F}_{\alpha ,q}$, by
combining the results for odd and even functions
(or using again the arguments in~\cite{KoorSw, KoorSw-corr}),
it is easy to check that
$\mathcal{F}_{\alpha ,q}^{-1}f(y) = \mathcal{F}_{\alpha ,q}f(-y)$.
Moreover, we have the formula
\begin{equation*}
\int_{-\infty }^{\infty } u(y)
\mathcal{F}_{\alpha ,q}v(y)\,d\mu _{q,\alpha}(y)
= \int_{-\infty }^{\infty }\mathcal{F}_{\alpha ,q}u(y)v(y)
\,d\mu_{q,\alpha }(y),
\end{equation*}
and $\mathcal{F}_{\alpha ,q}$ is an isometry on
$L^{2}(\mathbb{R},d\mu_{q,\alpha })$.
As in the case of the Dunkl transform, we can consider the
parameter $\alpha >-1$.

\subsection{The space $\mathcal{P}$}
\label{sub:calP}

The space $\mathcal{P}$ is the following $q$-analogue of the Paley-Wiener
space:
\begin{equation*}
PW_{\alpha ,q} = \Bigg\{f\in L^{2}(\mathbb{R},d\mu _{q,\alpha }):
f(t) = \int_{-1}^{1}u(x)E_{\alpha }(ixt;q^{2})\,d\mu _{q,\alpha }(x),
\; u\in L^{2}(\mathbb{R},d\mu _{q,\alpha })\Bigg\}.
\end{equation*}

\subsection{The biorthogonal functions}

Let us start defining the generalized little $q$-Gegenbauer polynomials
and, later, we will take the corresponding $q$-Fourier-Neumann type series.

\subsubsection{Generalized little $q$-Gegenbauer polynomials}

To construct the plane wave expansion for the kernel~\eqref{eq:qDunkl},
$q$-analogue of the Dunkl transform, we consider generalized little
$q$-Gegenbauer polynomials
\begin{align*}
C_{2n}^{(\beta +1/2,\alpha +1/2)}(t;q^{2})
&= (-1)^{n} \, \frac{(q^{2\alpha+2\beta +2};q^{2})_{n}}
{(q^{2\alpha +2};q^{2})_{n}}
\,p_{n}^{(\alpha ,\beta )}(t^{2};q^{2}), \\
C_{2n+1}^{(\beta +1/2,\alpha +1/2)}(t;q^{2})
&= (-1)^{n} \, \frac{(q^{2\alpha+2\beta +2};q^{2})_{n+1}}
{(q^{2\alpha +2};q^{2})_{n+1}}
\,t p_{n}^{(\alpha+1,\beta )}(t^{2};q^{2}),
\end{align*}
where the polynomials $p_{n}^{(\alpha ,\beta )}$ are defined
by~\eqref{eq:pnlittlejacobi} in terms of the little $q$-Jacobi polynomials.
Using~\eqref{eq:ortqjacob} we obtain
\begin{equation*}
\int_{-1}^{1} C_{k}^{(\beta +1/2,\alpha +1/2)}(t;q^{2})
C_{j}^{(\beta +1/2,\alpha +1/2)}(t;q^{2})
\frac{(t^{2}q^{2};q^{2})_{\infty }}{(t^{2}q^{2\beta
+2};q^{2})_{\infty }} \,d\mu _{q,\alpha }(t)
= h_{k,q}^{(\beta,\alpha)} \delta _{k,j},
\end{equation*}
where
\begin{align*}
h_{2n,q}^{(\beta ,\alpha )} &= \int_{-1}^{1}
\left[ C_{2n}^{(\beta+1/2,\alpha+1/2)}(t;q^{2})\right] ^{2}
\frac{(t^{2}q^{2};q^{2})_{\infty }}
{(t^{2}q^{2\beta +2};q^{2})_{\infty }}
\,d\mu _{q,\alpha }(t)
\\
&= \frac{1}{1-q^{2\alpha +2\beta+4n +2}}
\frac{(q^{2\alpha +2\beta+2};q^{2})_{n}} {(q^{2\alpha +2};q^{2})_{n}}
\frac{(q^{2n+2},q^{2\alpha+2\beta +2};q^{2})_{\infty }}
{(q^{2\beta +2n+2},q^2;q^{2})_{\infty }},
\end{align*}
\begin{align*}
h_{2n+1,q}^{(\beta ,\alpha )} &= \int_{-1}^{1} \left[
C_{2n+1}^{(\beta +1/2,\alpha +1/2)}(t;q^{2})\right] ^{2}
\frac{(t^{2}q^{2};q^{2})_{\infty }}
{(t^{2}q^{2\beta +2};q^{2})_{\infty }}
\,d\mu _{q,\alpha }(t)
\\
&= \frac{1}{1-q^{2\alpha +2\beta+4n +4}}
\frac{(q^{2\alpha+2\beta+2};q^{2})_{n+1}} {(q^{2\alpha +2};q^{2})_{n+1}}
\frac{(q^{2n+2},q^{2\alpha +2\beta +2};q^{2})_{\infty}}
{(q^{2\beta +2n+2},q^2;q^{2})_{\infty }}.
\end{align*}
We will also consider the little $q$-Gegenbauer polynomials defined as
\begin{equation}
C_{n}^{\beta }(t;q^{2}) = C_{n}^{(\beta ,0)}(t;q^{2})  \label{eq:qultra}
\end{equation}
(which can also be expressed in terms of big $q$-Jacobi
polynomials, as we can see in \cite[formulas (4.48) and (4.49)]{Koor}).

\subsubsection{$q$-Fourier-Neumann type series}

Now, given $\alpha >-1$, we define the $q$-Neumann functions by
\begin{equation*}
\mathcal{J}_{\alpha,n}(x;q^2)
= \frac{J_{\alpha+n+1} (xq^{[\frac{n+1}{2}]};q^{2})}{x^{\alpha+1}},
\end{equation*}
where $[\frac{n+1}{2}]$ denotes the biggest integer less or equal than
$\frac{n+1}{2}$. The identity
\begin{multline}  \label{eq:qweber2}
\int_{0}^{\infty} x^{-\lambda }
J_{\mu }(q^{m}x;q^{2}) J_{\nu}(q^{n}x;q^{2})
\, d_{q}x \\
= \begin{cases}
 (1-q)q^{n(\lambda -1)+(m-n)\mu }
\dfrac{(q^{1+\lambda+\nu-\mu},q^{2\mu +2};q^{2})_{\infty}}
{(q^{1-\lambda +\nu +\mu},q^{2};q^{2})_{\infty}} \\
\kern20pt\times {_{2}\phi_{1}} \left(\,
\begin{matrix}
q^{1-\lambda +\mu +\nu },q^{1-\lambda+\mu-\nu} \\
q^{2\mu +2}
\end{matrix}
\,\middle|\, q^{2};q^{2m-2n+1+\lambda+\nu-\mu} \,\right), \\
(1-q)q^{m(\lambda -1)+(n-m)\nu }
\dfrac{(q^{1+\lambda+\mu-\nu},q^{2\nu+2};q^{2})_{\infty}}
{(q^{1-\lambda +\mu +\nu},q^{2};q^{2})_{\infty}} \\
\kern20pt\times {_{2}\phi_{1}} \left(\,
\begin{matrix}
q^{1-\lambda +\nu +\mu },q^{1-\lambda+\nu-\mu} \\
q^{2\nu +2}
\end{matrix}
\,\middle|\, q^{2};q^{2n-2m+1+\lambda+\mu-\nu} \,\right),
\end{cases}
\end{multline}
was established in \cite{KoorSw, KoorSw-corr}, and it is valid for
$\Re\lambda <\Re(\mu +\nu +1)$, $m$ and $n$ integers. It can be
checked, by using Heine's transformation formula
\begin{equation*}
{_{2}\phi_{1}} \left(\,
\begin{matrix}
a,b \\
c
\end{matrix}
\,\middle|\, q;z\,\right) =
\frac{(b,az;q)_\infty}{(c,z,q)_{\infty}} \, {_{2}\phi_{1}}
\left(\,
\begin{matrix}
c/b,z \\
az
\end{matrix}
\,\middle|\, q;b\,\right),
\end{equation*}
that the expressions given on the right-hand side
of~\eqref{eq:qweber2} are equal; but there are some exceptional
cases in the previous identity (these exceptional cases were
overlooked in \cite{KoorSw} and they can be seen
in~\cite{KoorSw-corr}, which is a corrected version of the first
paper): the integral is only equal to the first part of the
right-hand side when $n-m+(1+\lambda+\mu-\nu)/2$ and
$(1-\lambda+\nu-\mu)/2$ are non-positive integers, and it is only
equal to the second part when $m-n+(1+\lambda+\nu-\mu)/2$ and
$(1-\lambda+\mu-\nu)/2$ are non-positive integers.

From~\eqref{eq:qweber2} we can state the following lemma.

\begin{Lem}
Let $\alpha>-1$. Then
\begin{equation*}
\int_{0}^{\infty} J_{\alpha+2n+1}(q^{n}x;q^{2})
J_{\alpha+2m+1}(q^{m}x;q^{2}) \, \frac{d_{q}x}{x}
= \frac{1-q}{1-q^{2\alpha+4m+2}} \,\delta_{n,m},
\end{equation*}
for $n,m=0,1,2,\dots$.
\end{Lem}

\begin{proof}
It is easy to check from~\eqref{eq:qweber2} that in the case
$q^n=q^m$, $\lambda=1$, and $\mu=\nu$,
\begin{equation*}
\int_{0}^{\infty} (J_{\mu }(q^{m}x;q^{2}))^2
\, \frac{d_{q}x}{x} \\
= \frac{1-q}{1-q^{2\mu}}.
\end{equation*}
For the case $n\not=m$, by setting $\lambda =1$, $\nu =\alpha+2n+1$
and $\mu=\alpha+2m+1$ in~\eqref{eq:qweber2}, it is clear that
\begin{multline*}
\int_{0}^{\infty} J_{\alpha+2n+1}(q^{n}x;q^{2})
J_{\alpha+2m+1}(q^{m}x;q^{2}) \, \frac{d_{q}x}{x}
= (1-q)q^{(m-n)(\alpha+2n+1)} \\
\times \frac{(q^{2n-2m+2},q^{2\alpha+4m+4};q^2)_{\infty}}
{(q^{2\alpha+2n+2m+2},q^2;q^2)}
\, {_{2}\phi_{1}} \left(\,
\begin{matrix}
q^{2\alpha+2n+2m+2},q^{2m-2n} \\
q^{2\alpha+4m+4}
\end{matrix}
\,\middle|\, q^2;q^2\,\right).
\end{multline*}
Then, by using the identity (see~\cite[formula (II.6)]{GR})
\begin{equation*}
{_{2}\phi_{1}} \left(\,
\begin{matrix}
a,q^{-n} \\
c
\end{matrix}
\,\middle|\, q;q\,\right) =\frac{(c/a;q)_n}{(c;q)_n}a^n
\end{equation*}
we deduce that the integral is null, and thus the proof is complete.
\end{proof}

By using the previous lemma with $\alpha$ and $\alpha+1$, and
taking into account that $\mathcal{J}_{\alpha,n}(x;q)$ is even or
odd according $n$ is even or odd, respectively, we have that
$\{\mathcal{J}_{\alpha,n}(x;q)\}_{n\ge 0}$ is an orthogonal system
on $L^{2}(\mathbb{R},d\mu_{\alpha,q}(x))$, namely
\begin{equation*}
\int_{-\infty}^{\infty} \mathcal{J}_{\alpha,n}(x;q^{2})
\mathcal{J}_{\alpha,m}(x;q^{2})\,d\mu_{q,\alpha}(x)
= \frac{(q^{2\alpha+2};q^{2})_{\infty}}{(q^{2};q^{2})_{\infty}}
\frac{1}{1-q^{2\alpha+2m+2}} \, \delta_{n,m},
\end{equation*}
for $n,m=0,1,2,\dots$.

To find the functions required in the ingredient (iii) of our method,
we consider
\begin{align}
\mathcal{Q}_{n}^{(\alpha ,\beta )}(t;q^{2})
&= (h_{n,q}^{(\beta ,\alpha )})^{-1}
\, \frac{(t^{2}q^{2};q^{2})_{\infty }}
{(t^{2}q^{2+2\beta};q^{2})_{\infty }}
C_{n}^{(\beta +1/2,\alpha +1/2)}(t;q^{2}),
\label{eq:calqQ} \\
\mathcal{P}_{n}^{(\alpha ,\beta )}(t;q^{2})
&= C_{n}^{(\beta +1/2,\alpha +1/2)}(t;q^{2}),
\label{eq:calqP}
\end{align}
and use the following lemma.

\begin{Lem}
\label{lem:qF-PQ} Let $\alpha, \beta > -1$, $\alpha+\beta > -1$,
and $k=0,1,2,\dots$. Then
\begin{equation}  \label{eq:qF-Q}
\mathcal{F}_{\alpha,q}(\mathcal{J}_{\alpha+\beta,k}(\,\cdot\,;q^{2}))(t)
= \frac{ (-i)^{k} q^{[\frac{k}{2}]
\beta}}{1-q^{2\alpha+2\beta+2k+2}}
\frac{(q^{2\alpha+2\beta+2};q^{2})_{\infty}}{(q^{2};q^{2})_{\infty}}
\mathcal{Q}_{k}^{(\alpha,\beta)}(t;q^{2}),
\end{equation}
for $t\in \{\pm q^{k}\}_{k\in\mathbb{Z}}$, and
\begin{equation}  \label{eq:qF-P}
\mathcal{F}_{\alpha,q}(| \cdot | ^{2\beta}
\mathcal{J}_{\alpha+\beta,k} (\,\cdot\,;q^{2}))(t) =
q^{-[\frac{k}{2}] \beta}(-i)^{k}
\frac{(q^{2\alpha+2};q^{2})_{\infty}} {(q^{2\alpha+2\beta+2
};q^{2})_{\infty}} \mathcal{P}_{k}^{(\alpha,\beta)}(t;q^2),
\end{equation}
for $t\in \{\pm q^k\}_{k\in \mathbb{Z}} \cap [-1,1]$.
\end{Lem}

The proof of Lemma~\ref{lem:qF-PQ} is contained in
subsection~\ref{sub:qprooflemma}.

\subsection{Main result}

\begin{Thm}
\label{thm:qDunkl} Let $\alpha,\beta > -1$ and $\alpha+\beta>-1$. Then for
each $x\in \{\pm q^k\}_{k\in \mathbb{Z}}$ the following expansion holds in
$L^{2}([-1,1],d\mu_{q,\alpha})$:
\begin{multline}
E_{\alpha}(ixt;q^{2}) \\
= \frac{(q^{2};q^{2})_{\infty}}{(q^{2\alpha+2\beta+2 };q^{2})_{\infty}}
\sum_{n=0}^{\infty} i^{n}
q^{-[\frac{n+1}{2}]\beta} (1-q^{2\alpha+2\beta+2n+2})
\mathcal{J}_{\alpha+\beta,n}(x;q^{2})
C_{n}^{(\beta+1/2,\alpha+1/2)}(t;q^{2}).  \label{eq:JA-qDunkl}
\end{multline}
Moreover, for $f \in PW_{\alpha,q}$, we have the orthogonal
expansion
\begin{equation*}
f(x) = \sum_{n=0}^\infty a_n(f) (1-q^{2\alpha +2\beta+2n +2})
\mathcal{J}_{\alpha+\beta,n}(x;q^2)
\end{equation*}
with
\begin{equation}  \label{eq:cnf-qDunkl}
a_n(f) = \frac{(q^{2};q^{2})_{\infty}} {(q^{2\alpha+2\beta+2
};q^{2})_{\infty}} \int_{\mathbb{R}} f(t)
\mathcal{J}_{\alpha+\beta,n}(t;q^2) \,d\mu_{q,\alpha+\beta}(t).
\end{equation}
Furthermore, the series converges uniformly in compact subsets
of~$\mathbb{R}$.
\end{Thm}

\begin{proof}
We proceed as in the proof of Theorem~2 in \cite{ACVExp} by using the
appropriate modifications. In the biorthogonal setup given in
section~\ref{sec:BBE}, let $\Omega =\mathbb{R}$, $I=[-1,1]$,
the space $L^{2}(I,d\mu )=L^{2}([-1,1],d\mu _{q,\alpha })$,
and the kernel $K(x,t)=E_{\alpha }(ixt;q^{2})$,
so $\mathcal{K}$ becomes~$\mathcal{F}_{\alpha ,q}$,
the $q$-analogue of the Dunkl transform defined in~\eqref{eq:qDu-T}
(and $\widetilde{K}=\mathcal{F}_{\alpha ,q}^{-1}$). Also, consider the
Paley-Wiener space $\mathcal{P}=PW_{\alpha ,q}$ of subsection~\ref{sub:calP}.
Finally, for $N=\mathbb{N}\cup \{0\}$, take the biorthonormal system given by
$P_{n}(t)=\mathcal{P}_{n}^{(\alpha ,\beta )}(t;q^{2})$
and $Q_{n}(t)=\mathcal{Q}_{n}^{(\alpha ,\beta )}(t;q^{2})$
as in~\eqref{eq:calqP} and~\eqref{eq:calqQ}. The result now follows
easily from Theorem~\ref{thm:expbilin} and Lemma~\ref{lem:qF-PQ}.
\end{proof}

\begin{Rem}
Taking the even parts in the identity \eqref{eq:JA-qDunkl}, we
deduced the following expansion for the kernel of the $q$-Hankel
transform:
\begin{multline*}
\frac{J_\alpha(xt;q^2)}{(xt)^\alpha}
= \frac{(q^{2\alpha+2};q^2)_\infty}{(q^{2\alpha+2\beta+2};q^2)_\infty}\\
\times \sum_{n=0}^\infty
q^{-n\beta}(1-q^{2\alpha+2\beta+4n+2})
\frac{(q^{2\alpha+2\beta+2};q^2)_n}{(q^{2\alpha+2};q^2)_n}
\frac{J_{\alpha+\beta+2n+1}(xq^n;q^2)}{x^{\alpha+\beta+1}}
p_n^{(\alpha,\beta)}(t^2;q^2),
\end{multline*}
valid for $\alpha$ and $\beta$ that satisfy $\alpha,\beta > -1$
and $\alpha+\beta>-1$.
Moreover, the functions belonging to the $q$-Hankel analogue of
the Paley-Wiener space, which is the domain of the sampling
theorem in~\cite{Abr-qs}, can be spanned by systems of $q$-Neumann
functions.
\end{Rem}

\subsection{Proof of formula \protect\eqref{eq:qlinearplanewave}}

Setting $\alpha =-\frac{1}{2}$ and replacing $\beta $ by
$\beta -\frac{1}{2}$, in~\eqref{eq:cnf-qDunkl} we obtain
the expansion for the $q$-exponential
function studied in~\cite{Rubin} in terms of the little $q$-Gegenbauer
polynomials defined in~\eqref{eq:qultra}:
\begin{equation*}
e(ixt;q^{2})
= \frac{(q^{2};q^{2})_{\infty }}{(q^{2\beta };q^{2})_{\infty }}
\,x^{-\beta } \sum_{n=0}^{\infty } i^{n}
q^{-[\frac{n+1}{2}](\beta-\frac{1}{2})} (1-q^{2\beta+2n })
J_{\beta +n}(xq^{[\frac{n+1}{2}]};q^{2})C_{n}^{\beta }(t;q^{2}).
\end{equation*}

\section{Technical lemmas}

In this section we present the calculations which provide us with the
$q$-analogues of the results in \cite[section~5]{ACVExp}.

\subsection{Some integrals involving $q$-Bessel functions}
\label{sub:qintBessel}

We will use the transformation
\begin{equation}  \label{eq:basictransform}
{_{2}\phi_{1}} \left(\,
\begin{matrix}
a,b \\
c
\end{matrix}
\,\middle|\, q;z \,\right)
= \frac{(abz/c;q)_{\infty}}{(z;q)_{\infty}}
\,{_{2}\phi_{1}} \left(\,
\begin{matrix}
c/a,c/b \\
c
\end{matrix}
\,\middle|\, q;abz/c \,\right);
\end{equation}
this formula appears in~\cite[formula (12.5.3)]{Ism}
subject to the conditions $|z|<1$ and $|abz| < |c|$,
but these restrictions on $z$ and on the parameters can be eliminated
because, by analytic continuation, the identity holds on $\mathbb{C}$
for all parameters as an identity of meromorphic functions
(se also \cite[p.~117]{GR}).
We also make repeated use of the obvious identity
$(a;q)_{\infty}=(a;q)_{n}(aq^{n};q)_{\infty}$.

\begin{Lem}
\label{lem:qHTjnab} For $\alpha,\beta >-1$ with $\alpha+\beta>-1$, and
$n=0,1,2,\dots$, let us define
\begin{equation*}
I_{-}(\alpha,\beta,n)(t,q) = \frac{t^{-\alpha }}{1-q}
\int_{0}^{\infty} x^{-\beta} J_{\alpha}(xt;q^{2})
J_{\alpha+\beta+2n+1}(q^{n}x;q^{2}) \, d_{q}x
\end{equation*}
and
\begin{equation*}
I_{+}(\alpha,\beta,n)(t,q)
= \frac{t^{-\alpha }}{1-q} \int_{0}^{\infty}
x^{\beta}J_{\alpha }(xt;q^{2})
J_{\alpha+\beta+2n+1}(q^{n}x;q^{2}) \, d_{q}x.
\end{equation*}
Then, we have
\begin{equation}  \label{eq:qHTjnab}
I_{-}(\alpha,\beta,n)(t,q) = q^{n\beta } \,
\frac{(q^{2\beta+2n+2};q^{2})_{\infty}}{(q^{2n+2};q^{2})_{\infty}}
\frac{(t^{2}q^{2};q^{2})_{\infty}}{(t^{2}q^{2\beta+2};q^{2})_{\infty}}
\, p_{n}^{(\alpha,\beta)}(t^{2};q^{2}),
\quad t\in \{q^{m}\}_{m\in \mathbb{Z}},
\end{equation}
and
\begin{equation}  \label{eq:qHTbjnab}
I_{+}(\alpha,\beta,n)(t,q) = q^{-n\beta } \,
\frac{(q^{2\alpha+n+2};q^{2})_{\infty}}
{(q^{2\alpha+2\beta+2n+2 };q^{2})_{\infty}}
\, p_{n}^{(\alpha,\beta)}(t^{2};q^{2}), \quad
t\in \{q^m\}_{m\in \mathbb{Z}} \cap (0,1].
\end{equation}
\end{Lem}

\begin{Rem}
Note that $I_{-}(\alpha,\beta,n)(t,q)=0$ for $t>1$. This is due to
the factor $(t^2q^2;q^2)_\infty$ involved in the formula.
\end{Rem}

\begin{Rem}
The identity \eqref{eq:qHTjnab} can be interpreted in terms of the
$q$-Hankel transform in the way
\begin{equation*}
H_{\alpha,q}(\mathcal{J}_{\alpha+\beta,2n}(\cdot;q^2))(t)
= q^{n\beta } \,
\frac{(q^{2\beta+2n+2};q^{2})_{\infty}}{(q^{2n+2};q^{2})_{\infty}}
\frac{(t^{2}q^{2};q^{2})_{\infty}}{(t^{2}q^{2\beta+2};q^{2})_{\infty}}
\, p_{n}^{(\alpha,\beta)}(t^{2};q^{2})
\end{equation*}
and, as a consequence of the inversion formula, it is also
verified that
\begin{equation*}
H_{\alpha,q}\left(q^{n\beta } \,
\frac{(q^{2\beta+2n+2};q^{2})_{\infty}}{(q^{2n+2};q^{2})_{\infty}}
\frac{((\cdot)^{2}q^{2};q^{2})_{\infty}}
{((\cdot)^{2}q^{2\beta+2};q^{2})_{\infty}}
\, p_{n}^{(\alpha,\beta)}((\cdot)^{2};q^{2})\right)(t)
= \mathcal{J}_{\alpha+\beta,2n}(t;q^2).
\end{equation*}
\end{Rem}

\begin{proof}[Proof of Lemma~\ref{lem:qHTjnab}]
We start evaluating $I_{-}(\alpha,\beta,n)(t,q)$ for $t\in
\{q^{m}\}_{m\in \mathbb{Z}}$. To this end, we take in
\eqref{eq:qweber2} $q^{m}=t$, $\mu =\alpha$,
$\nu=\alpha+\beta+2n+1$ and $\lambda =\beta$. For $t\le 1$ or
$t>1$ and $\beta$ non-integer we can use the first part of the
right-hand side of \eqref{eq:qweber2} to compute
$I_{-}(\alpha,\beta,n)(t,q)$ because we are not in the exceptional
situations. Then, in these cases,
\begin{multline*}
I_{-}(\alpha,\beta,n)(t,q) \\
= q^{n(\beta -\alpha -1)}
\frac{(q^{2\beta+2n+2},q^{2\alpha+2};q^{2})_{\infty}}
{(q^{2\alpha+2n+2 },q^{2};q^{2})_{\infty}} {\,{}_{2}\phi_{1}}
\left(\,
\begin{matrix}
q^{2\alpha+2n+2},q^{-2n-2\beta} \\
q^{2\alpha+2}
\end{matrix}
\,\middle|\, q^{2};t^{2}q^{2\beta+2} \,\right).
\end{multline*}
Moreover, the previous identity can be extended to all $\beta$ if
$t>1$ by continuity of the integral $I_{-}(\alpha,\beta,n)(t,q)$.
Now, applying formula~\eqref{eq:basictransform} and the definition
of $p_n^{(\alpha,\beta)}$ in terms of the little $q$-Jacobi
polynomials, we have
\begin{multline*}
I_{-}(\alpha,\beta,n)(t,q) \\
\begin{aligned}
&= q^{n(\beta-\alpha-1)} \,
\frac{(q^{2\beta+2n+2},q^{2\alpha+2};q^{2})_{\infty}}
{(q^{2\alpha+2n+2},q^{2};q^{2})_{\infty}} \\*
&\kern20pt \times
\frac{(t^{2}q^{2};q^{2})_{\infty}}{(t^{2}q^{2\beta+2};q^{2})_{\infty}}
\,{{}_{2}\phi_{1}} \left(\,
\begin{matrix}
q^{-2n},q^{2\alpha+2\beta+2n+2}
\\ q^{2\alpha+2}
\end{matrix}
\,\middle|\, q^{2};t^{2}q^{2}
\,\right) \\
&= q^{n(\beta -\alpha -1)} \,
\frac{(q^{2\beta+2n+2},q^{2\alpha+2};q^{2})_{\infty}}
{(q^{2\alpha +2n+2},q^{2};q^{2})_{\infty}}
\frac{(t^{2}q^{2};q^{2})_{\infty}}
{(t^{2}q^{2\beta+2};q^{2})_{\infty}} \,
p_{n}(t^{2};q^{2\alpha };q^{2\beta};q^{2}) \\
&= q^{n\beta} \, \frac{(q^{2\beta+2n+2};q^{2})_{\infty}}
{(q^{2n+2};q^{2})_{\infty}} \frac{(t^{2}q^{2};q^{2})_{\infty}}
{(t^{2}q^{2\beta+2};q^{2})_{\infty}} \,
p_{n}^{(\alpha,\beta)}(t^{2};q^{2}),
\end{aligned}
\end{multline*}
and the proof of~\eqref{eq:qHTjnab} is completed.

To prove the second part of the lemma, for $t=q^m\in(0,1]$ we
evaluate the integral by considering the first part of the
right-hand side of~\eqref{eq:qweber2} and choosing the parameters
$\mu =\alpha$, $\nu =\alpha+\beta+2n+1$ and $\lambda =-\beta$.
(The first part of the right-hand side of~\eqref{eq:qweber2} cannot
be used for $t>1$, because we are then in the exceptional
case. An expression for $I_{+}(\alpha,\beta,n)(t,q)$ valid for all
$t$ could be obtained with the second part of the right-hand side
of~\eqref{eq:qweber2}.) This results in
\begin{multline*}
\begin{aligned}
I_{+}(\alpha,\beta,n)(t,q) &= q^{-n(\beta+\alpha+1)}
\frac{(q^{2n+2},q^{2\alpha+2};q^{2})_{\infty}}
{(q^{2\alpha+2\beta+2n+2 },q^{2};q^{2})_{\infty}} \\*
&\kern20pt
\times{{}_{2}\phi_{1}} \left(\,
\begin{matrix} q^{-2n},q^{2\alpha+2\beta+2n+2} \\
q^{2\alpha+2} \end{matrix}
\,\middle|\, q^{2};t^{2}q^{2} \,\right) \\
&= q^{-n(\beta+\alpha+1)}
\frac{(q^{2n+2},q^{2\alpha+2};q^{2})_{\infty}}
{(q^{2\alpha+2\beta+2n+2},q^{2};q^{2})_{\infty}}
\, p_{n}(t^{2};q^{2\alpha };q^{2\beta };q^{2}) \\
&= q^{-n\beta}\frac{(q^{2\alpha+n+2};q^{2})_{\infty}}
{(q^{2\alpha+2\beta+2n+2};q^{2})_{\infty}} \,
p_{n}^{(\alpha,\beta)}(t^{2};q^{2}).
\end{aligned}
\end{multline*}
In this manner, we have proved~\eqref{eq:qHTbjnab} and the proof of the
lemma is finished.
\end{proof}

\subsection{Proof of Lemma~\protect\ref{lem:qF-PQ}}
\label{sub:qprooflemma}

Let us analyze the case $k=2n$ for \eqref{eq:qF-Q}. By decomposing on even
and odd functions we can write
\begin{equation}  \label{eq:FJqfirst}
\mathcal{F}_{\alpha,q}(\mathcal{J}_{\alpha+\beta,2n}(\,\cdot\,;q^{2}))(t)
= \frac{1}{1-q} \int_{0}^{\infty}
\frac{J_{\alpha+\beta+2n+1}(q^{n}x;q^{2})}{x^{\alpha+\beta+1}}
\frac{J_{\alpha}(xt;q^{2})}{(xt)^{\alpha}} \, x^{2\alpha+1} \, d_{q}x.
\end{equation}
Then, for $t>0$, $\alpha,\beta >-1$, and $\alpha+\beta >-1$, by
using~\eqref{eq:qHTjnab}, it is verified that
\begin{multline*}
\begin{aligned}
\mathcal{F}_{\alpha,q}(\mathcal{J}_{\alpha+\beta,2n}(\,\cdot\,;q^{2}))(t)
&= q^{n\beta } \,
\frac{(q^{2\beta+2n+2};q^{2})_{\infty}}{(q^{2n+2};q^{2})_{\infty}}
\frac{(t^{2}q^{2};q^{2})_{\infty}}{(t^{2}q^{2\beta+2};q^{2})_{\infty}}
\,
p_{n}^{(\alpha,\beta)}(t^{2};q^{2}) \\
&= (-1)^{n}q^{n\beta} \,
\frac{(q^{2\alpha+2};q^{2})_{n}}{(q^{2\alpha+2\beta+2};q^{2})_{n}}
\frac{(q^{2\beta+2n+2};q^{2})_{\infty}}
{(q^{2n+2};q^{2})_{\infty}} \\*
&\kern20pt \times\frac{(t^{2}q^{2};q^{2})_{\infty}}
{(t^{2}q^{2\beta+2};q^{2})_{\infty}}
C_{2n}^{(\beta+1/2,\alpha+1/2)}(t;q^{2}) \\
&= \frac{(-1)^{n}q^{n\beta }}{1-q^{2\alpha+2\beta+4n+2}}
\frac{(q^{2\alpha+2\beta+2};q^{2})_{\infty}}
{(q^{2};q^{2})_{\infty}}
\mathcal{Q}_{2n}^{(\alpha,\beta)}(t;q^{2}) .
\end{aligned}
\end{multline*}
For $t<0$, let us make in \eqref{eq:FJqfirst} the change $t_{1}=-t$, use the
evenness of the function $J_{\alpha}(z)/z^{\alpha}$, proceed as in the case
$t>0$, and undo the change. Then, for $k=2n$, we get
\begin{equation*}
\mathcal{F}_{\alpha,q}(\mathcal{J}_{\alpha+\beta,k}(\,\cdot\,;q^{2}))(t)
= \frac{(-i)^{k}q^{\frac{k}{2}\beta }}{1-q^{2\alpha+2\beta+2k+2}}
\frac{(q^{2\alpha+2\beta+2};q^{2})_{\infty}}{(q^{2};q^{2})_{\infty}}
\mathcal{Q}_{k}^{(\alpha,\beta)}(t;q^{2}) \chi_{[-1,1]}(t).
\end{equation*}
The case $k=2n+1$ works in a similar way.

Now, we are going to prove \eqref{eq:qF-P}. Again let us analyze the case
$k=2n$. By decomposing on even and odd functions we can write
\begin{equation*}
\mathcal{F}_{\alpha,q}(|\cdot|^{2\beta }
\mathcal{J}_{\alpha+\beta,2n}(\,\cdot\,;q^{2}))(t)
= \frac{1}{1-q} \int_{0}^{\infty} x^{2\beta }
\frac{J_{\alpha+\beta+2n+1}(q^{n}x;q^{2})} {x^{\alpha+\beta+1}}
\frac{J_{\alpha}(xt;q^{2})}{(xt)^{\alpha}}x^{2\alpha+1} \, d_{q}x.
\end{equation*}
Then, if $0<t<1$, we can use~\eqref{eq:qHTbjnab} to obtain
\begin{align*}
\mathcal{F}_{\alpha,q}(|\cdot| ^{2\beta }
\mathcal{J}_{\alpha+\beta,2n}(\,\cdot\,;q^{2}))(t) &= q^{-n\beta }
\frac{(q^{2\alpha+n+2};q^{2})_{\infty}}
{(q^{2\alpha+2\beta+2n+2};q^{2})_{\infty}}
\, p_{n}^{(\alpha,\beta)}(x^{2};q^{2}) \\
&= (-1)^n q^{-n\beta }\frac{(q^{2\alpha+2};q^{2})_{\infty}}
{(q^{2\alpha+2\beta+2 };q^{2})_{\infty}}
C_{2n}^{(\beta+1/2,\alpha+1/2)}(t;q^{2}) \\
&= (-1)^n q^{-n\beta }\frac{(q^{2\alpha+2};q^{2})_{\infty}}
{(q^{2\alpha+2\beta+2 };q^{2})_{\infty}}
\mathcal{P}_{2n}^{(\alpha,\beta)}(t;q^{2}).
\end{align*}
For $t<0$ we proceed as in the previous identity. Then, for $k=2n$, we get
\begin{equation*}
\mathcal{F}_{\alpha,q}(|\cdot|^{2\beta }
\mathcal{J}_{\alpha+\beta,k}(\,\cdot\,;q^{2}))(t)
= q^{-\frac{k}{2}\beta} (-i)^{k}
\frac{(q^{2\alpha+2};q^{2})_{\infty}}
{(q^{2\alpha+2\beta+2};q^{2})_{\infty}}
\mathcal{P}_{k}^{(\alpha,\beta)}(t;q^2).
\end{equation*}
The case $k=2n+1$ can be checked with the same arguments.

\subsection*{Acknowledgements}
We thank the referees their
extremely careful reading of the previous versions of
this paper and their useful suggestions,
that have allowed to correct some points and to
considerably improve the final version.
In particular, the explanations of one of the
referees about \cite{KoorSw} and \cite{KoorSw-corr}
have been very useful.



\end{document}